\documentclass[12pt]{elsarticle}
\usepackage{amssymb}
\usepackage{amsthm}
\usepackage{amsmath}
\usepackage{latexsym}
\usepackage{times}
\usepackage{graphicx}

\usepackage[active]{srcltx}

\newtheorem{Theorem}{Theorem}

\newtheorem{Lemma}{Lemma}

\newcommand{\clo}{\overline}

\begin{document}

\begin{frontmatter}




\title{\textbf{The DtN nonreflecting boundary condition for multiple scattering problems in the half-plane}}

\author[]{Sebastian Acosta \corref{cor1}}
\ead{sebastian@math.byu.edu}

\author[]{Vianey Villamizar}
\ead{vianey@math.byu.edu}

\author[]{Bruce Malone}
%\ead{}

\address{Department of Mathematics, Brigham Young University, Provo,
UT 84602, United States}

\cortext[cor1]{Corresponding author.
Tel: 1 801 361 9833.}

\begin{abstract}
The multiple-Dirichlet-to-Neumann (multiple-DtN) non-reflecting boundary condition is adapted to acoustic scattering from obstacles embedded in the half-plane. The multiple-DtN map is coupled with the method of images as an alternative model for multiple acoustic scattering in the presence of acoustically soft and hard plane boundaries. As opposed to the current practice of enclosing all obstacles with a large semicircular artificial boundary that contains portion of the plane boundary, the proposed technique uses small artificial circular boundaries that only enclose the immediate vicinity of each obstacle in the half-plane. The adapted multiple-DtN condition is simultaneously imposed in each of the artificial circular boundaries. As a result the computational effort is significantly reduced. A computationally advantageous boundary value problem is numerically solved with a finite difference method supported on boundary-fitted grids. Approximate solutions to problems involving two scatterers of arbitrary geometry are presented. The proposed numerical method is validated by comparing the approximate and exact far-field patterns for the scattering from a single and from two circular
obstacles in the half-plane.
\end{abstract}


\end{frontmatter}

% \linenumbers

%%%%%%%%%%%%%%%%%%%%%%%%%%%%%%%%%%%%%%%%%%%%%%%%%%%%%%%%%%%%%%%%%%%%
%%   S E C T I O N
%%%%%%%%%%%%%%%%%%%%%%%%%%%%%%%%%%%%%%%%%%%%%%%%%%%%%%%%%%%%%%%%%%%%
\section{Introduction} \label{Intro}

Wave scattering emerges in many applications such as sonar, radar,
antennas, seismic exploration, crack detection, satellite imaging,
and microscopy. Many of these scattering problems are more
realistically modeled by taking into account the presence of
infinite boundaries such as the ground surface in outdoor acoustics, radar and satellite imaging \cite{BertYoung01,CompAtmAc}, or the surface and bottom of the ocean in marine acoustics \cite{MarineAc,Fawcett02,UrickBook1983}. The analytical solutions for acoustic scattering in the half-plane or half-space can be found
only when the target conforms to simple shapes such as circles in
two dimensions, or spheres in three dimensions. These solutions can be constructed using the method of images so that the problem extends to multiple scattering in the full-space or full-plane. Then, modal expansions in separable coordinates systems or explicit evaluations of integral representations are available. For details, see \cite{BertYoung01,Twer1952}, \cite[Appendix A]{LeeKalli01}, or the book by Martin \cite{MartinBook} with its extensive reference list.

For arbitrarily shaped obstacles, a closed form solution is not
generally found, although some useful asymptotic approximations have
been constructed using integral representations in the low frequency
regime \cite{MartinBook,DassKlein01,DassKlein}. Hence, the scattering from an
obstacle of complex geometry embedded in the half-space generally
requires the application of numerical methods. One class of such methods is 
based on boundary integral formulations. We direct the reader to \cite[Chapters 5 and 6]{MartinBook} for an excellent overview and reference list of integral equation methods employed in multiple scattering theory. Some of the well-known advantages of these methods are the reduction in dimensionality (from volume to surface) of the domain of discretization, the automatic satisfaction of the radiation condition at infinity and of the boundary condition on the plane boundary for semi-infinite media. This is accomplished by using the correct Green's function as the integral kernel. However, these methods may become quite costly since they lead to dense matrices and the reduction of dimensionality is lost if the media contains inhomogeneities. Another important class of numerical methods for scattering problems are based on finite elements and finite differences. These volume discretization techniques, as opposed to boundary integral methods, are naturally suited for treating localized heterogeneities, nonlinearities and sources, as well as the presence of impenetrable obstacles. Their major drawback, however, is the handling of the unboundedness of the domain.

A great deal of research has been performed to formulate
appropriate absorbing boundary conditions in order to truncate the
physical domain. These boundary conditions should allow the outgoing
waves to leave the truncated domain without nonphysical
reflections. Some of these conditions include local absorbing
boundary conditions \cite{Bayliss01,Engquist01,
Kriegsmann87,Hagstrom98}, perfectly matched layers
\cite{Berenger01,Bermudez07,Collino98,Turkel-Yefet98}
and an exact nonreflecting boundary condition known as
Dirichlet-to-Neumann (DtN) map \cite{Keller01,Grote-Keller01,Givoli1999}. The
potential advantages and drawbacks of these conditions have been
extensively studied by Givoli \cite{GivoliReview,GivoliReview2,GivoliBook} and
Tsynkov \cite{Tsynkov1998}. Some are easier to implement, others are
more robust, while others may be computationally more efficient.
However, a general agreement in favor of one of these conditions
over the others has not been reached.

For the full-plane, the most common practice has been to enclose all scatterers with a single artificial circular or elliptical boundary and apply the absorbing boundary condition on it. For half-plane problems,
multiple scattering interactions between the plane boundary and the scatterers take place. For that reason, the artificial boundary typically consists of a semi-circular or elliptical artificial boundary enclosing all scatterers and a portion of the plane boundary.
This was the approach followed by Lee and Kallivokas \cite{LeeKalli01}, and Givoli and Vigdergauz \cite{GivVig1993} where excellent numerical results for the scattering from two-dimensional obstacles were obtained. For instance, an effort was made in
\cite{LeeKalli01} to extend the applicability of elliptically shaped
absorbing boundaries to half-plane problems. Its use rendered significant computational
savings for elongated obstacles near the plane boundary in comparison with
the use of semi-circular boundaries. However, if the obstacle is relatively far from the plane boundary or if the scatterer consists of several disjoint obstacles, the use of a
single artificial boundary (circular or elliptical) to enclose all the obstacles and a portion of the plane boundary will inevitably lead to a large computational domain.

An improvement can be made in terms of reducing the size of the computational domain. In fact, Grote and Kirsch \cite {Grote01} introduced the multiple-DtN map as a nonreflecting boundary condition for the scattering from
several obstacles embedded in the full-space or full-plane. Later, Acosta and Villamizar \cite{JCP2010} combined it with curvilinear coordinates and applied it to obstacles of arbitrary shape. Its definition requires the introduction of separate artificial boundaries, each one enclosing a different obstacle, reducing the computational region to a set of small sub-domains. Then, the multiple-DtN condition is  simultaneously imposed  on each artificial
boundary. The net effect of this condition is that propagating waves are allowed to leave the computational sub-domains without spurious reflections and simultaneously account for the wave interactions between the different sub-domains. As mentioned in \cite{Grote01}, neither local absorbing boundary conditions nor
perfectly matched layers in their current form successfully deal with such multiple
scattering interactions.

The purpose of the present work is to adapt the  multiple-DtN technique, described above  for scattering in the full-plane, to scattering problems in the half-plane. The main idea
can be conveniently illustrated by Fig. \ref{First}. Here, two domains corresponding to equivalent problems are depicted. In Fig. \ref{First}(a), the semi-infinite domain $\Omega$ of the
{\it original half-plane problem}, internally bounded  by $\mathcal{C}$ and lying above the plane boundary $\Gamma$, is shown. On the other hand in Fig. \ref{First}(b), the small annular domain $\Omega_{\text {int}}$ of the final problem, bounded internally by $\mathcal{C}$ and externally by the artificial circular boundary $\mathcal{B}$, is also shown. The proposed method based on an adaptation of the multiple-DtN technique requires: First, the construction of a new multiple-DtN condition for the half-plane which handles the interactions between the plane boundary and the scatterer. Secondly, the approximation of the solution of the unbounded original problem by numerically solving the equivalent final problem in the bounded domain $\Omega_{\text{int}}$ with the novel multiple-DtN condition imposed on the artificial circular boundary $\mathcal{B}$. This final problem will be called {\it half-plane multiple-DtN scattering problem} in this work.
\begin{figure}[!ht]
\begin{center}
\includegraphics[width=0.9\textwidth]{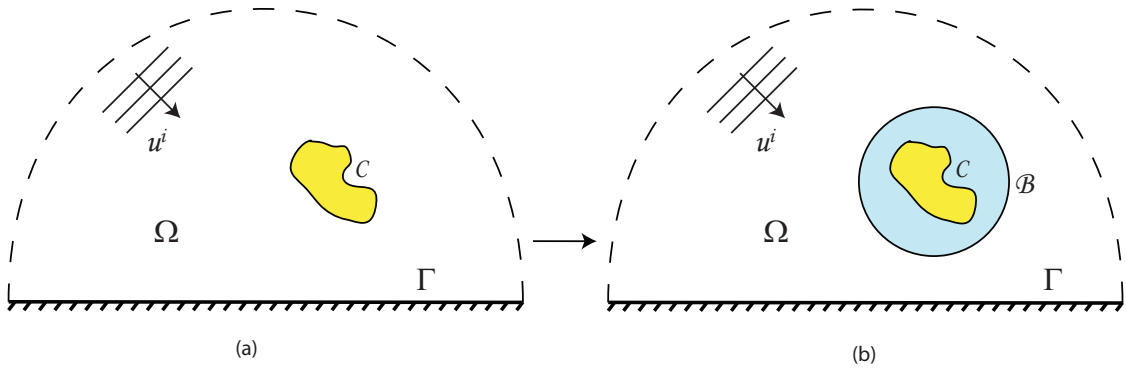}
\caption{(a) Original physical scattering problem, (b) Truncated domain where the new multiple-DtN condition will be imposed and the numerical approximations to the original problem will be calculated.}
\label{First}
\end{center}
\end{figure}

The derivation of the adapted condition involves several steps. To start,  consider the physical problem of scattering from a single obstacle in the half-plane with an acoustically soft or hard boundary condition on the plane boundary $\Gamma$. A natural approach
consists of using the \emph{method of images} to extend the half-plane problem to a multiple scattering problem containing two scatterers in the full-plane. This problem is in turn reduced to an equivalent two-obstacle bounded BVP in the full-plane by using the multiple-DtN map \cite{Grote01}, as explained above. The intrinsic symmetry of the method of images is exploited for a further reduction of the two-obstacle bounded problem to a single-obstacle BVP defined in the small bounded domain $\Omega_{\text {int}}$ of the half-plane, as shown in Fig. \ref{First}(b). The symmetry also leads to a natural derivation of the new multiple-DtN condition for the half-plane from the corresponding multiple-DtN condition in the full-plane. The derivation of this condition and its extension to a finite number of obstacles in the half-plane constitute one of the main contribution of this work.

The final BVP with the novel multiple-DtN boundary condition, defined in the small domain $\Omega_{\text {int}}$ of the half-plane, can be accurate and efficiently solved by using appropriate  numerical volume methods such as finite elements or finite differences. The radius of the exterior artificial circular boundary $\mathcal{B}$ of the small domain $\Omega_{\text {int}}$ is independent of the distance between the obstacle and the plane boundary. Since the scattering interactions between the plane boundary and the obstacle are handled by the new multiple-DtN condition, there is no need to retain the interacting portion of the plane boundary within the truncated domain. As a result, the computational region is significantly reduced in comparison with the use of a large semi-circular artificial boundary enclosing not only the obstacle, but also a portion of the plane boundary. This is significant in the simulation of sonar and radar problems. For instance, the distance between a submarine and the bottom of the ocean may be several times larger than the characteristic length of the submarine.

The multiple-DtN boundary condition is independent of the numerical
discretization employed in the truncated domain. It has been
successfully implemented using finite element (FEM) and finite
difference (FDM) methods, as seen in \cite{Grote01,JCP2010} and references
therein. In the present work, we follow the approach found in
\cite{JCP2010} for obstacles of arbitrary shape. The problem is written in terms of generalized
curvilinear coordinates inside the small computational region $\Omega_{\text{int}}$, that may consists of a finite number of small sub-domains. Then, novel
elliptic grids conforming to complexly shaped two-dimensional
obstacles are constructed and approximations of the scattered field
supported by them are obtained using a FDM. Nevertheless, the FEM
can also be applied as done in other applications of the DtN map
\cite{Keller01,Grote-Keller01,Givoli1999,Grote01,GivKeller1989,GivoliBook}.

The proposed numerical method is validated by comparing the
approximate and exact far-field patterns for a configuration consisting of one and two
circular obstacles embedded in the upper-half-plane. In particular, by employing a second order finite difference scheme, a second order convergence of the
numerical solution to the exact solution is easily verified.

%%%%%%%%%%%%%%%%%%%%%%%%%%%%%%%%%%%%%%%%%%%%%%%%%%%%%%%%%%%%%%%%%%%%
%%   S E C T I O N
%%%%%%%%%%%%%%%%%%%%%%%%%%%%%%%%%%%%%%%%%%%%%%%%%%%%%%%%%%%%%%%%%%%%

\section{Mathematical formulation} \label{MathModel}
The mathematical formulation of two-dimensional acoustic scattering by obstacles embedded in the upper half-plane is considered.
To simplify, only scattering from a single acoustically soft obstacle in a homogeneous isotropic medium is studied in this section. The generalization to a finite number of obstacles is discussed in Section \ref{MultiObst}.
The scatterer consists of a simply connected region bounded by a closed
smooth curve $\mathcal{C}$. A cartesian coordinate system where $\textbf{x}=(x,y)$ denotes a point in $\mathbb{R}^2$ is employed. These coordinates are chosen so that the $x$-axis agrees with the physical straight line $\Gamma$ which will be called the {\it plane boundary}. The acoustic medium $\Omega$ is the multiply connected open region of the upper-plane bounded from below by $\Gamma$ and internally by $\mathcal{C}$, as shown in Fig. \ref{First}(a).

An incident plane wave, $ u^{i}(\textbf{x})= e^{i k \textbf{x} \cdot \textbf{d}}$, propagating in $\Omega$ is impinging upon the obstacle and the plane boundary $\Gamma$. The vector $\textbf{d}=(d_x,d_y)$ is a unit vector that points in the {direction of
incidence,} $i=\sqrt{-1}$ and $k>0$ represents the wave number. The reflected wave $u^r$ generated when $u^i$ scatterers from $\Gamma$ in the absence of any obstacles is also considered. For a rigid plane boundary $\Gamma$,
$u^r({\textbf{x}})= e^{i k \textbf{x} \cdot {\tilde {\bf d}}}$ while for soft $\Gamma$, $u^r({\textbf{x}})= -e^{i k \textbf{x} \cdot {\tilde {\textbf{d}}}}$, where
$\tilde {\textbf{d}}=(d_x,-d_y)$.

The total field $u^t$ is
decomposed as $u^t = u^{i} + u^r + u^{s}$ in $\clo{\Omega}$ where $u^{s}$ represents the
scattered field induced by the presence of the obstacle. Since $u^i$ and $u^r$ satisfy the Helmholtz equation and their normal derivatives cancel each other over $\Gamma$, the scattering problem for a soft obstacle and a rigid boundary $\Gamma$ consists
of finding the scattered field $u^{s}$ satisfying,
\begin{eqnarray}
&& \Delta u^{s} + k^2 u^{s} = 0 \quad\qquad \text{in $\Omega$}, \label{BVPsc1} \\
&& u^{s} = -u^{i}-u^r \qquad\qquad \text{on $\mathcal{C}$,} \label{BVPsc2} \\
&&\frac{\partial{u^{s}}}{\partial y} = 0 \qquad\qquad\qquad \text{on $\Gamma$,} \label{BVPsc3}\\
&& \lim_{\rho \rightarrow \infty} \rho^{1/2} \left( \frac{\partial u^s} {\partial \rho} - i k u^{s} \right) = 0. \label{BVPsc4}
\end{eqnarray}
The limit in (\ref{BVPsc4}), known as
the Sommerfeld radiation condition, is assumed to hold uniformly for
all directions above $\Gamma$, and $\rho = |\textbf{x}|$.
The analysis for other physically interesting cases, such as when the plane boundary is soft, or the obstacle is rigid are very similar. Numerical results combining these cases for one or multiple obstacles are also presented in later sections.

Since the boundary condition (\ref{BVPsc3}) on the unbounded boundary
$\Gamma$ is homogeneous, it is possible to show
existence and uniqueness of a classical solution for (\ref{BVPsc1})-(\ref{BVPsc4}) by mimicking the classical  approach described in \cite[Chapter 3]{ColtonKress02}
for obstacles embedded in the full-plane. The proof of uniqueness is
based on Rellich's lemma adapted for the half-plane scenario. Details on Rellich's lemma for domains with infinite boundaries can be found in \cite{PeStoker1954,Miran1957,Odeh1963}. Existence is based on surface potentials and the Riesz-Fredholm theory for integral equations of the second
kind. The only major change required is the use of the upper
half-plane Green's function instead of the two-dimensional fundamental
solution. The Green's functions for the Neumann (rigid)
and Dirichlet (soft) conditions on $\Gamma$ are easily
constructed using the well-known method of images.

The main goal of the present work is the formulation of an efficient
nonreflecting boundary condition that \emph{simultaneously} accounts for the unboundedness of the domain $\Omega$ and the presence of the plane boundary $\Gamma$. This is the subject of the following sections.

%%%%%%%%%%%%%%%%%%%%%%%%%%%%%%%%%%%%%%%%%%%%%%%%%%%%%%%%%%%%%%%%%%%%
%%   S E C T I O N
%%%%%%%%%%%%%%%%%%%%%%%%%%%%%%%%%%%%%%%%%%%%%%%%%%%%%%%%%%%%%%%%%%%%

\section{Equivalent full-plane multiple-scattering problem and the multiple-DtN map} \label{SHCond}
The original scattering problem in the half-plane (\ref{BVPsc1})-(\ref{BVPsc4}) defined in the previous section is now extended to an equivalent full-plane multiple scattering problem by means of the well-known \emph{method of images}, as shown in Fig. \ref{Second}(a).
\begin{figure}[!ht]
\begin{center}
\includegraphics[width=0.9 \textwidth]{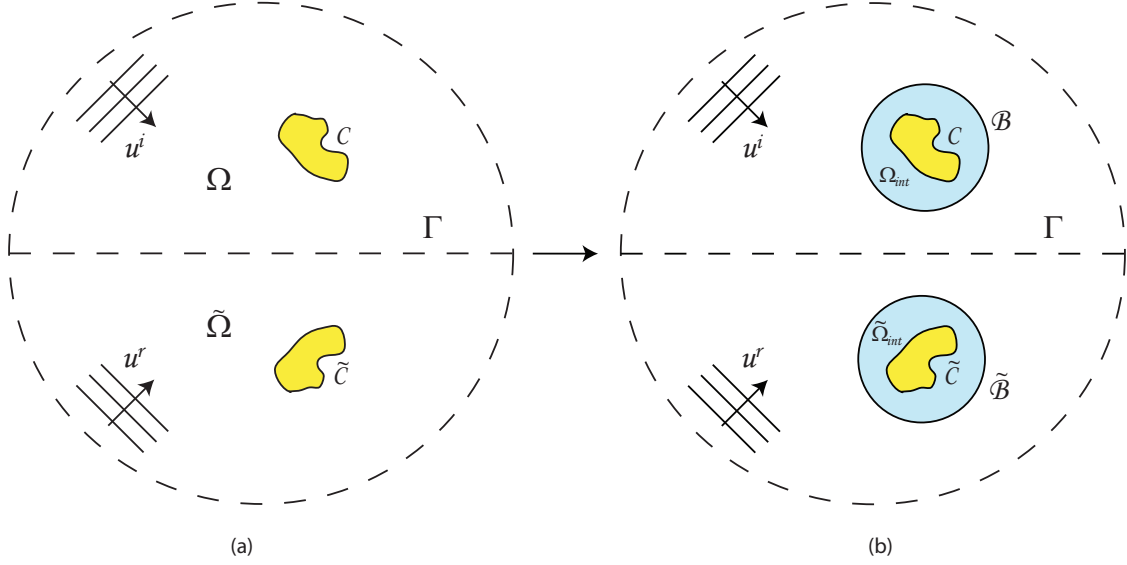}
\caption{(a) Extended full-plane multiple scattering problem, (b) Full-plane multiple-DtN scattering problem.}
\label{Second}
\end{center}
\end{figure}

As a first step, a mirror-image obstacle (with respect to $\Gamma$) with boundary $\mathcal{\tilde C}$ is defined in the lower half-plane. The new bounding curve $\mathcal{\tilde C}$ of the mirror obstacle is formed by all points $(x,y)$ in the lower half-plane such that $(x,-y)\in {\mathcal{C}}$. The image region of $\Omega$ in the lower half-plane is analogously defined and is denoted as
$\tilde{\Omega}$. As a consequence, the domain for the extended problem consists of $\Omega\cup\tilde{\Omega}\cup\Gamma,$ and $\Gamma$ is no longer a boundary.
A new BVP for the full-plane can be defined now. It consists of  finding $u^{s}$ satisfying,
\begin{eqnarray}
&& \Delta u^{s} + k^2 u^{s} = 0 \qquad  \text{in $\Omega \cup \Gamma \cup \tilde{\Omega}$}, \label{BVPfull1} \\
&& u^{s} = -u^{i}-u^r \qquad\qquad \text{on $\mathcal{C} \cup \tilde{\mathcal{C}}$,} \label{BVPfull2} \\
&& \lim_{\rho \rightarrow \infty} \rho^{1/2} \left( \frac{\partial u^s} {\partial \rho} - i k u^{s} \right) = 0, \label{BVPfull3}
\end{eqnarray}
where $u^{i}$ and $u^{r}$ are defined as in Section \ref{MathModel}, but with their domain in the full-plane $\mathbb{R}^2$.
The existence and uniqueness for this BVP is proven in \cite[Chapter 3]{ColtonKress02}.

It is also easy to show that if $u^s(x,y)$, with $(x,y)\in\clo{{\Omega \cup \Gamma \cup \tilde{\Omega}}}$ is a solution of the BVP (\ref{BVPfull1})-(\ref{BVPfull3}), then ${\hat u}^s(x,y)=u^s(x,-y)$, is also a solution. Then, uniqueness and the definition of $\Omega$, $\tilde\Omega$, and $\Gamma$ leads to
\begin{equation}
u^s(x,y)=u^s(x,-y),
\qquad\mbox{for all}\quad
(x,y)\in\clo{{\Omega\cup\Gamma\cup\tilde{\Omega}}}
\end{equation}
Since this property plays a key role in this work, we will express it as the following lemma.
\begin{Lemma} \label{Lemma1}
The solution $u^s$ of the BVP (\ref{BVPfull1})-(\ref{BVPfull3}) is symmetric about the plane boundary $\Gamma$.
\end{Lemma}
A trivial consequence of Lemma \ref{Lemma1} is that $(\partial u^s / \partial y)(x,0)=0$. This result can be used to establish the equivalence between BVPs  (\ref{BVPsc1})-(\ref{BVPsc4}) and (\ref{BVPfull1})-(\ref{BVPfull3}).
The advantage of extending the original single scattering half-plane problem to the equivalent full-plane multiple scattering problem for two mirror obstacles is that a computationally efficient and accurate numerical method (multiple-DtN technique) for this problem has recently been  obtained by Grote and Kirsch \cite {Grote01}.
The novel multiple-DtN technique consists of reducing the unbounded domain of the above full-plane BVP into small annular regions surrounding the obstacles, as illustrated in Fig. \ref{Second}(b). This is followed by simultaneously imposing on each artificial
boundary an appropriate nonreflecting boundary condition that avoids spurious reflection and handles the multiple obstacles interactions. This condition is called the multiple-DtN boundary condition.

More precisely following \cite{JCP2010}, we assume that the obstacle in the upper half-plane is well separated from $\Gamma$. Therefore, it can be enclosed by an artificial circular boundary
$\mathcal{B}$ with radius $R$, centered at a point $\textbf{b}=(b_x,b_y)$ inside the  obstacle, which is completely contained in $\Omega$ ($b_y>R$). Similarly, an image circle $\tilde{\mathcal{B}}$ of $\mathcal{B}$ with center at $\tilde{\textbf{b}}=(b_x,-b_y)$ and same radius $R$, completely contained in $\tilde{\Omega}$ and enclosing the image obstacle, can be constructed. As a result of this procedure, two small annular bounded domains $\Omega_{\text {int}}\subset\Omega$, bounded internally by  $\mathcal{C}$ and externally by
$\mathcal{B}$,  and $\tilde{\Omega}_{\text{int}}\subset\tilde{\Omega}$, bounded internally by $\tilde{\mathcal{C}}$ and externally by the artificial circular boundary
$\tilde{\mathcal{B}}$, are obtained. These two small domains are clearly symmetric about $\Gamma$. Also, $\Psi$ and $\tilde{\Psi}$ denote the infinite open
regions, bounded internally by $\mathcal{B}$ and
$\tilde{\mathcal{B}}$, respectively.
The common unbounded region outside the two circular artificial boundaries $\mathcal{B}$ and $\tilde{\mathcal{B}}$ is denoted as
  $\Omega_{\text{ext}}=
\Psi\cap\tilde{\Psi}$.

Since the artificial boundary
$\mathcal{B}$ is a circle, it is convenient to define a local polar coordinate system $(r,\theta)$, with origin at the center
${\textbf{b}}=( b_x, b_y)$ of the circle ${\mathcal{B}},$ in the outer region $\Psi$. Similarly, another local polar coordinate system
$(\tilde{r},\tilde{\theta}),$ with origin at the center
$\tilde{\textbf{b}}$ of the circle $\tilde{\mathcal{B}},$ is defined in the unbounded region $\tilde{\Psi}$. Since the angles $\theta$ and
$\tilde{\theta}$ are counterclockwise measured,
the symmetric image of the point $\textbf{x} = (R,\theta) \in \mathcal{B}$ about the $x$-axis is the point
$\tilde{\textbf{x}}=(R,\tilde{\theta}) \in \tilde{\mathcal{B}},$ with $\tilde{\theta}=-\theta$.

The derivation of the full-plane multiple-DtN map is based upon the following important {\it Decomposition Theorem} previously formulated and proved in Proposition 1 in
\cite{Grote01} and Theorem 2 in \cite{JCP2010}.

\begin{Theorem} \label{TheoremDecomp}
Let $u^{s}$ be the solution to the BVP
(\ref{BVPfull1})-(\ref{BVPfull3}). Then, $u^{s}$ can be uniquely
decomposed into purely outgoing wave fields $v$ and $\tilde{v}$ such
that
\begin{equation}
u^{s} =  v + \tilde{v} ,\qquad \text{in
$\overline{\Omega}_{\text{ext}}$}, \label{Decomp1}
\end{equation}
where
\begin{eqnarray}
&\Delta v + k^2 v = 0 \qquad \text{in $\Psi$},
&\quad\mbox{and}\quad\lim_{r \rightarrow \infty} r^{1/2} (\frac{\partial v}{\partial r} - i k v ) = 0,\label{Decomp2}\\
&\Delta \tilde{v} + k^2 \tilde{v} = 0 \qquad \text{in  $\tilde{\Psi}$},
&\quad\mbox{and}\quad\lim_{\tilde{r} \rightarrow \infty} \tilde{r}^{1/2} ( \frac{\partial \tilde{v}}{\partial \tilde{r}} - i k
\tilde{v} ) = 0.
\label{Decomp3}
\end{eqnarray}
\end{Theorem}

Analytical formulas for $v$ and $\tilde{v}$ in terms of the boundary values $v(R,\theta')$ and $\tilde{v}(R,\tilde{\theta}')$ on the artificial circular boundaries $\mathcal{B}$ and $\tilde{\mathcal{B}},$ respectively, can be readily obtained using eigenfunction expansions. They are given by
\begin{eqnarray}
&&v(r,\theta) = \frac{1}{2\pi} \sum_{n=0}^{\infty}
\epsilon_{n} \frac{H_{n}(k r)}{H_{n}(k R)} \int_{0}^{2 \pi}
v(R,\theta')\cos n (\theta-\theta') d \theta', \label{vexp1}
\end{eqnarray}
\begin{eqnarray}
&&\tilde{v}(\tilde{r},\tilde{\theta}) = \frac{1}{2\pi}
\sum_{n=0}^{\infty} \epsilon_{n} \frac{H_{n}(k
\tilde{r})}{H_{n}(k R)} \int_{0}^{2 \pi} \tilde{v}(R,\tilde{\theta}')\cos n
(\tilde{\theta}- \tilde{\theta}') d \tilde{\theta'}, \label{vexp2}
\end{eqnarray}
where $r, \tilde{r}\ge R$, $0\le \theta,\tilde{\theta}\le 2\pi$, $H_{n}$ stands
for the $n^{th}$ order Hankel function of the first kind, and
$\epsilon_n$ is the Neumann factor, {i.e.,} $\epsilon_0 = 1$ and
$\epsilon_n = 2$ for $n \geq 1$.

Among the results found in \cite{Grote01,JCP2010}, as a consequence of Theorem \ref{TheoremDecomp}, one of great practical value is the possibility of reducing the extended full-plane multiple-scattering problem (\ref{BVPfull1})-(\ref{BVPfull3}) to the following equivalent full-plane bounded problem for the restriction of $u^s$ to $\Omega_{\text{int}} \cup \tilde{\Omega}_{\text{int}}$ and the values of $v$ and $\tilde{v}$ on $\mathcal{B}$ and $\mathcal{\tilde{B}}$, respectively (see Fig. \ref{Second}),
\begin{eqnarray}
&& \Delta u^{s} + k^2 u^{s} = 0 \qquad  \text{in $\Omega_{\text{int}} \cup \tilde{\Omega}_{\text{int}}$}, \label{BVPDtN1} \\
&& u^{s} = -u^{i} - u^{r} \qquad\qquad \text{on $\mathcal{C} \cup \tilde{\mathcal{C}}$,} \label{BVPDtN2} \\
&& u^{s} = v + \tilde{P}[\tilde{v}] \qquad\qquad \text{on $\mathcal{B}$,} \label{BVPDtN3} \\
&& u^{s} = \tilde{v} + P[v] \qquad\qquad \text{on $\tilde{\mathcal{B}}$,} \label{BVPDtN4} \\
&& \partial_{r}u^{s} = M[v] + \tilde{T}[\tilde{v}] \qquad \text{on $\mathcal{B}$,} \label{BVPDtN5} \\
&& \partial_{\tilde{r}}u^{s} = \tilde{M}[\tilde{v}] + T[v] \qquad \text{on $\mathcal{\tilde{B}}$.} \label{BVPDtN6}
\end{eqnarray}
In the remainder of this work, this BVP will be called {\it full-plane multiple-DtN scattering problem}.
Here we have adopted similar operators, $P$, ${\tilde P}$, $M$, ${\tilde M}$, $T$, and ${\tilde T}$, as those introduced by Grote and Kirsch in \cite{Grote01}.
As discussed in \cite{Grote01,JCP2010}, conditions (\ref{BVPDtN3})-(\ref{BVPDtN4}) describe the continuity of the scattered field across the artificial boundaries $\mathcal{B}$ and $\tilde{\mathcal{B}}$, respectively. Similarly, conditions (\ref{BVPDtN5})-(\ref{BVPDtN6}) correspond to the continuity of the normal derivative of the scattered field across those boundaries, respectively.

The above operators are defined by operating over the analytical formulas (\ref{vexp1})-(\ref{vexp2}) for $v$ and $\tilde{v}$, evaluated on points in either $\mathcal{B}$ or $\tilde{\mathcal{B}}$.
For instance, the operator ${\tilde P}$ maps values of the outgoing field $\tilde{v}|_{\tilde{\mathcal{B}}}$ to values of the the same outgoing field on $\mathcal{B}$. Its definition can be conveniently expressed as  ${\tilde P} : \tilde{v}|_{\tilde{\mathcal{B}}} \mapsto
\tilde{v}|_{\mathcal{B}}$. Similarly, the operator ${P}$ is defined as
${P} : v|_\mathcal{B} \mapsto {v}|_{\tilde{\mathcal{B}}}$.
These operators are called the \emph{propagation} operators in \cite{Grote01}. Explicit formulas can be obtained from the analytical expressions
(\ref{vexp1})-(\ref{vexp2}) for $v$ and $\tilde{v}$. In fact,
\begin{eqnarray}
&& P[v](\tilde{\theta}) = \frac{1}{2 \pi} \sum_{n=0}^{\infty} \epsilon_{n} \frac{H_{n}(k r(\tilde{\theta}))}{H_{n}(kR)} \int_{0}^{2 \pi} v(R,\theta') \cos n(\theta(\tilde{\theta}) - \theta') d \theta', \label{defP} \\
&& {\tilde P}[\tilde{v}](\theta) = \frac{1}{2 \pi} \sum_{n=0}^{\infty} \epsilon_{n} \frac{H_{n}(k \tilde{r}(\theta))}{H_{n}(kR)} \int_{0}^{2 \pi} \tilde{v}(R,\tilde{\theta}') \cos n(\tilde{\theta}(\theta) - \tilde{\theta}') d \tilde{\theta}'. \label{defPt}
\end{eqnarray}

The operators $M$ and ${\tilde M}$ correspond to the standard DtN operator. The first one maps the Dirichlet datum $v|_{\mathcal{B}}$  to the Neumann datum
$\partial_r v|_{\mathcal{B}}$ on ${\mathcal{B}}$,  i.e.,
$M: v|_{\mathcal{B}} \mapsto \partial_r v|_{\mathcal{B}}$, while ${\tilde M}:
\tilde{v}|_{\tilde{\mathcal{B}}}\mapsto\partial_ {\tilde{r}}
\tilde{v}|_{\tilde{\mathcal{B}}}$. They are explicitly given by
\begin{eqnarray}
&& M[v](\theta) = \frac{k}{2 \pi} \sum_{n=0}^{\infty} \epsilon_{n} \frac{H'_{n}(k R)}{H_{n}(kR)} \int_{0}^{2 \pi} v(R,\theta') \cos n(\theta - \theta') d \theta', \label{defM} \\
&& {\tilde M}[\tilde{v}](\tilde{\theta}) = \frac{k}{2 \pi} \sum_{n=0}^{\infty} \epsilon_{n} \frac{H'_{n}(k R)}{H_{n}(kR)} \int_{0}^{2 \pi} \tilde{v}(R,\tilde{\theta}') \cos n(\tilde{\theta} - \tilde{\theta}') d \tilde{\theta}'. \label{defMt}
\end{eqnarray}

Finally, the \emph{transfer} operators $T$ maps  $v |_{\mathcal{B}}$ to its derivative in the outer normal direction to the circle $\tilde{\mathcal{B}}$ at points on $\tilde{\mathcal{B}}$, i.e,
$T : v |_{\mathcal{B}} \mapsto \partial_{\tilde{r}} v|_{\tilde{\mathcal{B}}}.$ Similarly,
${\tilde T} : \tilde{v} |_{\tilde{\mathcal{B}}} \mapsto \partial_{r} \tilde{v}|_{\mathcal{B}}$. They are defined by the explicit formulas
\begin{eqnarray}
T[v](\tilde{\theta}) &=& \cos(\theta(\tilde{\theta}) - \tilde{\theta} ) \frac{k}{2 \pi} \sum_{n=0}^{\infty} \epsilon_{n} \frac{H'_{n}(k r(\tilde{\theta}))}{H_{n}(kR)} \int_{0}^{2 \pi} v(R,\theta') \cos n(\theta(\tilde{\theta}) - \theta') d \theta' \notag \\
&& \quad + \frac{\sin(\theta(\tilde{\theta}) - \tilde{\theta})}{2 \pi r(\tilde{\theta})} \sum_{n=0}^{\infty} \epsilon_{n} \frac{n H_{n}(k r(\tilde{\theta}))}{H_{n}(kR)} \int_{0}^{2 \pi} v(R,\theta') \sin n(\theta(\tilde{\theta}) - \theta') d \theta', \label{defT} \\
{\tilde T}[\tilde{v}](\theta) &=& \cos(\tilde{\theta}(\theta) - \theta) \frac{k}{2 \pi} \sum_{n=0}^{\infty} \epsilon_{n} \frac{H'_{n}(k \tilde{r}(\theta))}{H_{n}(kR)} \int_{0}^{2 \pi} \tilde{v}(R,\tilde{\theta}') \cos n(\tilde{\theta}(\theta) - \tilde{\theta}') d \tilde{\theta}' \notag \\
&& \quad + \frac{\sin(\tilde{\theta}(\theta) - \theta)}{2 \pi \tilde{r}({\theta})} \sum_{n=0}^{\infty} \epsilon_{n} \frac{n H_{n}(k \tilde{r}(\theta))}{H_{n}(kR)} \int_{0}^{2 \pi} \tilde{v}(R,\tilde{\theta}') \sin n(\tilde{\theta}(\theta) - \tilde{\theta}') d \tilde{\theta}'. \label{defTt}
\end{eqnarray}

For the evaluations ${\tilde T}[{v}]$ and ${\tilde P}[\tilde{v}]$, the dependence of $(\tilde{r},\tilde{\theta})$ on $(R,\theta)$ is needed. This is given by the following formulas,
\begin{eqnarray}
&& \tilde{r}^2 = 4 b_{y}^2 + R^2 + 4 b_{y}R \sin
\theta, \label{rel1} \\
&& \cos \tilde{\theta} = \frac{R}{\tilde{r}} \cos \theta, \qquad \sin \tilde{\theta} = \frac{R}{\tilde{r}} \sin \theta.
\label{rel2}
\end{eqnarray}
The analogous dependence of $(r,\theta)$ on $(R,\tilde{\theta})$ is also needed for the evaluations $T[v]$ and $P[v]$.

Adopting the techniques used in \cite{Grote01,JCP2010}, the full-plane multiple-DtN problem (\ref{BVPDtN1})-(\ref{BVPDtN6}) can be numerically solved and the scattered field $u^{s}$ can be approximated everywhere in $\Omega_{\text{int}} \cup \tilde{\Omega}_{\text{int}}$. Instead of doing this, we propose to further reduce this full-plane bounded BVP to an equivalent half-plane bounded BVP. This is accomplished by exploiting the symmetry referred to in Lemma \ref{Lemma1}. A similar procedure is found in \cite{GivVig1993,Givoli05,GivKeller1989,GivoliBook} for problems admitting geometrical symmetries in the DtN formulation. The practical importance of this process is that the computational effort is significantly reduced since only half of the compuational domain is to be discretized. This is the subject of the following sections.

%%%%%%%%%%%%%%%%%%%%%%%%%%%%%%%%%%%%%%%%%%%%%%%%%%%%%%%%%%%%%%%%%%%%
%%   S E C T I O N
%%%%%%%%%%%%%%%%%%%%%%%%%%%%%%%%%%%%%%%%%%%%%%%%%%%%%%%%%%%%%%%%%%%%

\section{The multiple-DtN map adapted to the half-plane} \label{MDtNhalf}
In this section, the full-plane multiple-DtN nonreflecting boundary condition is adapted to  half-plane scattering
problems with acoustically soft or hard plane boundaries. This constitutes the major contribution of the present work. The goal is to derive an exact nonreflecting boundary condition for the circular artificial boundary $\mathcal{B}$ that does not depend on the outgoing field $\tilde{v}$ explicitly. Then, by using this new boundary condition transform the extended full-plane multiple-DtN scattering BVP
(\ref{BVPDtN1})-(\ref{BVPDtN6}) into an equivalent {\it half-plane multiple-DtN scattering BVP}, as shown in Fig. \ref{Third}.
\begin{figure}[!ht]
\begin{center}
\includegraphics[width=0.9 \textwidth]{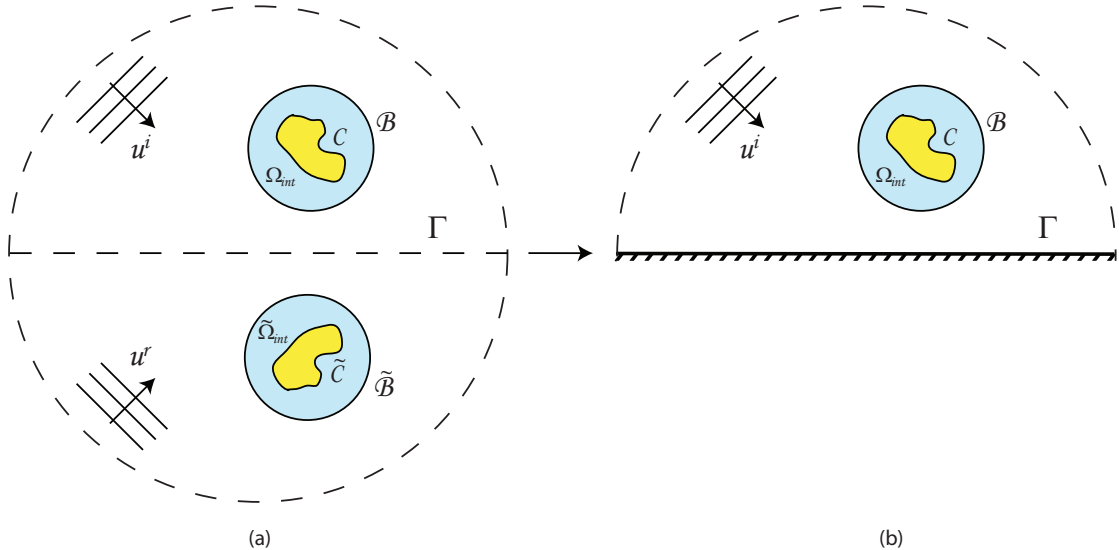}
\caption{(a) Full-plane multiple-DtN scattering problem, (b) Half-plane multiple-DtN scattering problem.}
\label{Third}
\end{center}
\end{figure}

As a first step to accomplish the above goal, we will establish an important symmetry relation between $v$ and $\tilde v$ in $\clo{\Omega}_{\text{ext}}$. This is the statement of the following lemma.
\begin{Lemma} \label{LemmaSymmetry}
Let $u^{s}$ be the solution to the BVP
(\ref{BVPfull1})-(\ref{BVPfull3}), and $v$ and $\tilde{v}$ be its unique decomposition into purely outgoing wave fields in $\clo{\Omega}_{\text{ext}}$ as stated in Theorem \ref{TheoremDecomp}. Then, $v(x,y) = \tilde{v}(x,-y)$ for all $(x,y) \in \clo{\Omega}_{\text{ext}}$.
\end{Lemma}

\begin{proof}
Define a pair $V$ and $\tilde V$ of outgoing fields as
\begin{eqnarray}
&& V(x,y) := \tilde{v}(x,-y), \qquad (x,y) \in \Psi, \label{NewWF1} \\
&& \tilde{V}(x,y) := v(x,-y), \qquad (x,y) \in \tilde{\Psi} \label{NewWF2},
\end{eqnarray}
where $v$ and $\tilde v$ form the unique decomposition of $u^s$ into purely outgoing waves of Theorem \ref{TheoremDecomp}. It immediately follows from these definitions that
$V$ and $\tilde{V}$ are also radiating solutions to the Helmholtz equation and satisfy
(\ref{Decomp2})-(\ref{Decomp3}).

It is also possible to show that   $V$ and $\tilde{V}$ as defined in (\ref{NewWF1})-(\ref{NewWF2}) constitute a decomposition of the scattered field $u^{s}$ of the form (\ref{Decomp1}). For this purpose, let $(x,y)$ be an arbitrary point in $\clo{\Omega}_{\text{ext}}$, then definitions (\ref{NewWF1})-(\ref{NewWF2}) and the decomposition (\ref{Decomp1}) lead to
\begin{eqnarray}
V(x,y) + \tilde{V}(x,y) = \tilde{v}(x,-y) + v(x,-y) = u^{s}(x,-y) = u^{s}(x,y), \qquad
(x,y) \in \clo{\Omega}_{\text{ext}}\notag
\end{eqnarray}
The last equality follows from the fact that $u^{s}$ is symmetric about $\Gamma$ as stated in Lemma \ref{Lemma1}. By the uniqueness of the decomposition of $u^{s}$ into outgoing wave fields (Theorem \ref{TheoremDecomp}), we conclude that
\begin{eqnarray}
&& v(x,y) = V(x,y)=\tilde{v}(x,-y), \qquad
(x,y) \in \clo{\Omega}_{\text{ext}}, \label{ImagD0}
\end{eqnarray}
which establishes the desired symmetry.
\end{proof}

Notice that the symmetric image of
$\tilde{\textbf{x}}=(R,\tilde{\theta}) \in \tilde{\mathcal{B}}$ about the $x$-axis is the point
$\textbf{x} = (R,-\theta) \in \mathcal{B},$ where $\tilde{\theta}=\theta$. Therefore, an immediate consequence of Lemma \ref{LemmaSymmetry} is that
\begin{eqnarray}
&& \tilde{v}(R,\tilde{\theta})=v(R,-\theta). \label{ImagD}
\end{eqnarray}
This is the key property for the definition of the new {\it half-plane multiple-DtN boundary condition} from the corresponding full-plane multiple-DtN condition (\ref{BVPDtN3})-(\ref{BVPDtN6}).

As a next step, we introduce a
\emph{symmetry} operator $S$. This operator $S$ maps a boundary value function $w$ defined in $\mathcal{B}$ to another boundary value function $\tilde w$ defined in $\tilde{\mathcal{B}}$. More precisely,
\begin{eqnarray}
&& S: w|_{\mathcal{B}}\mapsto \tilde{w}|_{\tilde{\mathcal{B}}}\notag\\
&& \tilde{w}(R,\tilde{\theta})=w(R,-\theta),\quad (R,\tilde{\theta}) \in \tilde{\mathcal{B}},\notag
\end{eqnarray}
where the numerical value of $\tilde\theta$ coincides with the numerical value of $\theta$, i.e., $\tilde\theta=\theta$.

Finally, the above results and the fact that $u^s$ is symmetric about $\Gamma$ (Lemma \ref{Lemma1}) suggest that the explicit computation of $u^s$ in $\tilde{\Omega}_{\text{int}}$ and $\tilde v$ on $\tilde{\mathcal{B}}$ in the full-plane problem are no longer needed. As a consequence, the
full-plane multiple-DtN BVP (\ref{BVPDtN1})-(\ref{BVPDtN6}) may be reduced to  the much simpler half-plane bounded BVP for $u$ in $\Omega_{\text{int}}$ and $w$ on $\mathcal{B}$ satisfying,
\begin{eqnarray}
&& \Delta u + k^2 u = 0 \qquad  \text{in $\Omega_{\text{int}}$}, \label{BVPDtNN1} \\
&& u = -u^{i} - u^{r} \qquad\qquad \text{on $\mathcal{C}$,} \label{BVPDtNN2} \\
&& u = w + {\tilde P}S[w] \qquad\qquad \text{on $\mathcal{B}$,} \label{BVPDtNN3} \\
&& \partial_{r} u = M[w] + {\tilde T}S [w]
\qquad \text{on $\mathcal{B}$,} \label{BVPDtNN4}
\end{eqnarray}
where $\tilde{P}S[w]=(\tilde{P}\circ S)[w]$ and $\tilde{T}S[w]=(\tilde{T}\circ S)[w]$. This is the BVP illustrated in Fig. \ref{Third}(b).
It will be shown below (Theorem \ref{TheoremEquiv}) that this BVP is well-posed and is indeed equivalent to the BVP (\ref{BVPDtN1})-(\ref{BVPDtN6}).  We call it the \emph{half-plane multiple-DtN scattering BVP}. Moreover,
the boundary conditions (\ref{BVPDtNN3})-(\ref{BVPDtNN4}) define the \emph{half-plane multiple-DtN boundary condition}.

\begin{Theorem} \label{TheoremEquiv}
The half-plane bounded scattering BVP (\ref{BVPDtNN1})-(\ref{BVPDtNN4}) has a unique solution $(u,w)$ such that $u$ coincides with the restriction of $u^{s}$ to $\Omega_{\text{int}}$, and $w$ coincides with $v$ on $\mathcal{B}$, where the triple $(u^{s},v,\tilde{v})$ is the unique solution to the full-plane bounded multiple scattering BVP (\ref{BVPDtN1})-(\ref{BVPDtN6}).
\end{Theorem}

\begin{proof}
It is clear that if the triple $(u^{s},v,\tilde{v})$ satisfies the BVP (\ref{BVPDtN1})-(\ref{BVPDtN6}), then $u^s$ satisfies
(\ref{BVPDtNN1})-(\ref{BVPDtNN2}). Moreover according to Lemma \ref{LemmaSymmetry}, $\tilde{v}(R,\tilde{\theta})=v(R,-\theta)$ or in other words $\tilde{v} = S[v]$. Thus by making this substitution into boundary conditions (\ref{BVPDtN3}) and (\ref{BVPDtN5}), they reduce to boundary conditions (\ref{BVPDtNN3}) and (\ref{BVPDtNN4}). Therefore,
if $(u^{s},v,\tilde{v})$  satisfies (\ref{BVPDtN3}) and (\ref{BVPDtN5}) , then $(u^{s},v)$ satisfies (\ref{BVPDtNN3}) and (\ref{BVPDtNN4}), respectively. Then, existence of solutions for the half-plane bounded problem is proven.

We will now prove, that a solution $(u,w)$ of (\ref{BVPDtNN1})-(\ref{BVPDtNN4}) can be extended to the unique solution of the BVP (\ref{BVPDtN1})-(\ref{BVPDtN6}). For this purpose, let us define
\begin{eqnarray}
&& U^{s}(x,y) := \left\{
                               \begin{array}{ll}
                                 u(x,y), & \hbox{if $(x,y) \in \clo{\Omega}_{\text{int}}$;} \\
                                 u(x,-y), & \hbox{if $(x,y) \in \clo{\tilde{\Omega}}_{\text{int}}$.}
                               \end{array}
                             \right. \label{CandidU}
\end{eqnarray}
and
\begin{eqnarray}
&& V(x,y) := w(x,y), \qquad (x,y) \in \mathcal{B}, \label{CandidV1} \\
&& \tilde{V}(x,y) := w(x,-y), \qquad (x,y) \in \tilde{\mathcal{B}}. \label{CandidV2}
\end{eqnarray}
Notice that $U^{s}$ satisfies the Helmholtz equation in $\Omega_{\text{int}} \cup \tilde{\Omega}_{\text{int}}$ because $u$ does so in $\Omega_{\text{int}}$. Additionally, $U^{s}$ satisfies the physical boundary condition on $\mathcal{C} \cup \tilde{\mathcal{C}}$ because of (\ref{BVPDtNN2}) and the fact that $u^{i}+u^{r}$ is symmetric about $\Gamma$.

The only part of the proof left is to check that the
full-plane multiple-DtN boundary conditions (\ref{BVPDtN3})-(\ref{BVPDtN6}) are satisfied by the triple $(U^{s},V,\tilde{V)}$.
Let us first consider the satisfaction of the boundary condition (\ref{BVPDtN3}). If $(R,\theta)\in\mathcal{B}$ then the definitions (\ref{CandidU})-(\ref{CandidV1}) and boundary condition (\ref{BVPDtNN3}) lead to
\begin{equation}
U^s(R,\theta) = u(R,\theta) = w(R,\theta) + \tilde{P}S[w](\theta) = V(R,\theta) + \tilde{P}S[w](\theta),\notag
\end{equation}
Thus, it is enough to show that $S[w]=\tilde{V}$, but this immediately follows from the definitions of $S$ and $\tilde{V}$ in (\ref{CandidV2}). The proof for the satisfaction
of boundary condition (\ref{BVPDtN5}) from boundary condition (\ref{BVPDtNN4}) is completely analogous.

For the verification of (\ref{BVPDtN4}) observe that if $(R,\tilde\theta)\in {\tilde{\mathcal{B}}}$ then, $(R,-\theta)\in {\mathcal{B}}$. Therefore, using boundary condition (\ref{BVPDtNN3}) and definitions (\ref{CandidU}) and  (\ref{CandidV2}) yields
\begin{equation}
U^s(R,\tilde{\theta})=u(R,-\theta)=w(R,-\theta)+{\tilde P}S[w](-\theta)={\tilde V}(R,\tilde\theta)+{\tilde P}S[w](-\theta).
\end{equation}
Thus, boundary condition (\ref{BVPDtN4}) would be verified if ${\tilde P}S[w](-\theta)=
 {P}[{V}](\tilde\theta)$ for $(R,\tilde\theta)\in {\tilde{\mathcal{B}}}$.
In fact,
\begin{equation}
{\tilde P}S[w](-\theta)=
\frac{1}{2 \pi} \sum_{n=0}^{\infty} \epsilon_{n} \frac{H_{n}(k \tilde{r}(-\theta))}{H_{n}(kR)} \int_{0}^{2 \pi} w(R,-{\theta}') \cos n(\tilde{\theta}(-\theta) - {\theta}') d{\theta}'.\label{last}
\end{equation}
Then, using definition (\ref{CandidV1})
and substituting the geometrical identities:
$\tilde{r}(-\theta)=r(\tilde{\theta})$ and ${\tilde{\theta}}(-\theta)=-\theta(\tilde\theta)$ into (\ref{last}), we arrive to the sought equality. The verification of boundary condition (\ref{BVPDtN6}) is completely analogous to (\ref{BVPDtN4}).

Hence, the triple $(U^{s},V,\tilde{V)}$ is in fact the unique solution to the multiple-DtN problem (\ref{BVPDtN1})-(\ref{BVPDtN6}). This means that $U^{s} \equiv u^{s}$, $V \equiv v$ and $\tilde{V} \equiv \tilde{v}$.
In summary, we have shown that the reduced problem (\ref{BVPDtNN1})-(\ref{BVPDtNN4}) has at least one solution. Furthermore, any solution of it can be extended to a solution of the multiple-DtN problem (\ref{BVPDtN1})-(\ref{BVPDtN6}), which in turn possesses a unique solution as shown in \cite[Thm. 2]{Grote01}. This forces the reduced problem to have exactly one solution $(u,w)$ such that $u$ coincides with the restriction of $u^{s}$ to $\Omega_{\text{int}}$, and $w$ coincides with $v$ on $\mathcal{B}$. This concludes the proof.
\end{proof}

The above results  can be easily modified to account for an acoustically soft plane boundary $\Gamma$. In fact, if the Dirichlet condition is imposed on $\Gamma$, the scattered field $u$ is anti-symmetric about $\Gamma$ so that the purely outgoing wave fields satisfy $v(x,y) = -\tilde{v}(x,-y)$ for $(x,y) \in \clo{\Omega}_{\text{ext}}$. Thus, the half-plane multiple-DtN  scattering problem in this case is given by
\begin{eqnarray}
&& \Delta u + k^2 u = 0 \qquad  \text{in $\Omega_{\text{int}}$}, \label{BVPDtND1} \\
&& u = -u^{i} - u^{r} \qquad\qquad \text{on $\mathcal{C}$,} \label{BVPDtND2} \\
&& u = w - {\tilde P}S[w] \qquad\qquad \text{on $\mathcal{B}$,} \label{BVPDtND3} \\
&& \partial_{r} u = M[w] - {\tilde T}S [w]
\qquad \text{on $\mathcal{B}$.} \label{BVPDtND4}
\end{eqnarray}

%%%%%%%%%%%%%%%%%%%%%%%%%%%%%%%%%%%%%%%%%%%%%%%%%%%%%%%%%%%%%%%%%%%%%
%%%   S E C T I O N
%%%%%%%%%%%%%%%%%%%%%%%%%%%%%%%%%%%%%%%%%%%%%%%%%%%%%%%%%%%%%%%%%%%%%
\section{The far-field pattern} \label{FFPSection}
By considering the asymptotic expansion of the scattered field at infinity, it is possible to know more about the scattered energy in the far-field. For a polar coordinate system $(\rho,\vartheta)$ with the origin at $(0,0),$ it is well-known that this expansion for a full-plane multiple scattering problem is given by
\begin{equation}
u^s(\rho,\vartheta) = \frac{e^{ik\rho}}{\sqrt{k\rho}}u_{\infty}(\vartheta)+\mathcal{O}(\rho^{-3/2}).
\label{lot}
\end{equation}
The angular dependent coefficient $u_{\infty}$ of the leading order term of this expansion is known as the far-field pattern. It represents how the intensity of the scattered energy  varies according to all possible $\vartheta$ directions.
As shown in \cite[Thm. 3]{Grote01} for the two-obstacle scattering problem in the full-plane, the far-field pattern can be explicitly obtained in terms of the outgoing wave fields $v$ and $\tilde{v}$. In fact, for two obstacles located symmetrically about $\Gamma$,
\begin{eqnarray}
u_{\infty}(\vartheta) = F[v](\vartheta) + \tilde{F}[\tilde{v}](\vartheta),\qquad \mbox{for $\vartheta\in [0,2\pi]$}  \label{FFP-Full}
\end{eqnarray}
where the operators $F$ and $\tilde{F}$ acting on $v$ and $\tilde{v}$, respectively, are given by
\begin{eqnarray}
&& F[v](\vartheta) = \frac{1-i}{2 \pi \sqrt{\pi}} e^{- i k (b_x \cos \vartheta + b_y \sin \vartheta)} \sum_{n=0}^{\infty} \epsilon_{n} \frac{(-i)^n }{H_{n}(k R)} \int_{0}^{2 \pi} v(R,\theta') \cos n(\vartheta - \theta') d \theta', \label{defF} \\
&& \tilde{F}[\tilde{v}](\vartheta) = \frac{1-i}{2 \pi \sqrt{\pi}} e^{- i k (b_x \cos \vartheta - b_{y} \sin \vartheta)} \sum_{n=0}^{\infty} \epsilon_{n} \frac{  (-i)^n}{H_{n}(k R)} \int_{0}^{2 \pi} \tilde{v}(R,\tilde{\theta}') \cos n(\vartheta - \tilde{\theta}') d \tilde{\theta}'. \label{defFt}
\end{eqnarray}
Recall that $\textbf{b}=(b_{x},b_{y})$ is the center of the circle $\mathcal{B},$ and $(b_{x},-b_{y})$ is its image about the axis $\Gamma$.
The formula (\ref{FFP-Full}) can be adapted to the half-plane problem by using the symmetry relation (\ref{ImagD}) between the outgoing wave fields $v$ and $\tilde{v}$. If the Neumann condition (\ref{BVPsc3}) is imposed on
 $\Gamma,$ then the far-field pattern for the half-plane problem can be written as
\begin{equation}
u_{\infty}(\vartheta) = F[v](\vartheta) + \tilde{F} S[v](\vartheta),\label{ff1}
\end{equation}
where again $S[v](R,\tilde\theta) = v(R,-\theta),$ for $\tilde\theta=\theta$. If instead the Dirichlet condition is assumed on $\Gamma,$ then the far-field pattern for the half-plane scattering problem can be written as $u_{\infty}(\vartheta) = F[v](\vartheta) - \tilde{F} S [v](\vartheta)$.

%%%%%%%%%%%%%%%%%%%%%%%%%%%%%%%%%%%%%%%%%%%%%%%%%%%%%%%%%%%%%%%%%%%%%
%%%   S E C T I O N
%%%%%%%%%%%%%%%%%%%%%%%%%%%%%%%%%%%%%%%%%%%%%%%%%%%%%%%%%%%%%%%%%%%%%
\section{Extension to multiple obstacles in the half-plane} \label{MultiObst}
In this section we re-formulate the reduced BVP (\ref{BVPDtNN1})-(\ref{BVPDtNN4}) to account for the presence of several obstacles in the upper half-plane. This reformulation is based on the extended version of Theorem \ref{TheoremDecomp} (Decomposition Theorem), as given in \cite[Thm. 2]{JCP2010}, and also on the straightforward generalization of Lemma \ref{LemmaSymmetry} and Theorem \ref{TheoremEquiv} to multiple scatterers.

To start, consider $J$ obstacles whose physical boundaries are denoted by $\mathcal{C}_{j}$ for $j=1,2,...,J$. Each one of the scatterers is enclosed by an artificial circular boundary $\mathcal{B}_{j}$ of radius $R_{j}$. The region exterior to this artificial boundary is denoted by $\Psi_{j}$, and the circles $\mathcal{B}_{j}$ are mutually disjoint. The immediate vicinity $\Omega_{j}$ of each obstacle is the region bounded internally by $\mathcal{C}_{j}$ and externally by $\mathcal{B}_{j}$. Hence, the computational domain is $\Omega_{\text{int}} = \cup_{j=1}^{J} \Omega_{j}$. As in Section \ref{SHCond}, the symbols  $\tilde{\mathcal{B}}_{j}$ and $\tilde{\Psi}_{j}$ denote the images about $\Gamma$ of the artificial boundary $\mathcal{B}_{j}$ and the unbounded domain $\Psi_{j}$, respectively. Finally, $\Omega_{\text{ext}} = \cap_{j=1}^{J} (\Psi_{j} \cap \tilde{\Psi}_{j})$ denotes the remaining unbounded exterior domain in the full-plane.
A polar coordinate system $(r_{j},\theta_{j})$ is defined in each of the regions $\clo{\Psi}_{j}$ and another polar coordinate system $(\tilde{r}_{j},\tilde{\theta}_{j})$ is defined in each of the regions $\clo{\tilde{\Psi}}_{j}$.

According to the extended version of Theorem \ref{TheoremDecomp},
the scattered field $u^{s}$ is uniquely decomposed into $2 J$ purely outgoing wave fields, in the exterior region $\Omega_{\text{ext}}$, i.e.,
\begin{eqnarray}
u^{s} = \sum_{j=1}^{J} v_{j} + \sum_{j=1}^{J} \tilde{v}_{j}, \quad \text{in $\clo{\Omega}_{\text{ext}}$}. \notag
\end{eqnarray}
Here, $\clo{\Psi}_{j}$ and $\clo{\tilde{\Psi}}_{j}$ are the domains of definition for $v_{j}$ and $\tilde{v}_{j}$, respectively.

Following the derivation of the previous sections, it can be easily proved that the  half-plane multiple scattering problem for hard plane boundaries is equivalent to the following {\it half-plane multiple-DtN  multiple scattering BVP}:
\begin{eqnarray}
&& \Delta u^{s} + k^2 u^{s} = 0 \qquad  \text{in $\Omega_{\text{int}}$}, \label{BVPDtNNM1} \\
&& u^{s} = -u^{i} - u^{r} \qquad\qquad \text{on $\mathcal{C}_{j}$\quad for $j=1,2,...,J$,} \label{BVPDtNNM2} \\
&& u^{s} = v_{j} + \sum_{l \neq j}^{J} P[v_{l}] + \sum_{l = 1}^{J} {\tilde P}S[v_{l}] \qquad\qquad \text{on $\mathcal{B}_{j}$ \quad for $j=1,2,...,J$,} \label{BVPDtNNM3} \\
&& \partial_{r} u^{s} = M[v_{j}] + \sum_{l \neq j}^{J} T[v_{l}] + \sum_{l = 1}^{J} {\tilde T}S [v_{l}]
\qquad \text{on $\mathcal{B}_{j}$ \quad for $j=1,2,...,J$.} \label{BVPDtNNM4}
\end{eqnarray}
Here, the operator $M[v_{j}](\theta_{j})$ is given by (\ref{defM}) with $v$ replaced by $v_{j}$, $\theta$ by $\theta_{j}$ and $R$ by $R_{j}$.
Similarly, $\tilde{P}S[v_{l}](\theta_j)$ and $\tilde {T}S[v_{l}](\theta_j)$ are given by (\ref{defPt}) and (\ref{defTt}) with $\tilde v(R,\tilde\theta)$ replaced by $S[v_{l}](R_l,\tilde\theta_l)=
v_l(R_l,-\theta_l)$, $\theta$ by $\theta_{j}$, $R$ by $R_l$, $\tilde{r}$ by $\tilde{r}_{l}$, and $\tilde{\theta}$ by $\tilde{\theta}_{l}$. The operators  $P[v_{l}](\theta_j$) and $T[v_{l}](\theta_j)$ are given by (\ref{defP}) and (\ref{defT}) with $\tilde v$ replaced by $v_{l}$, $\theta$ by $\theta_j$, $\tilde{\theta}$ by $\theta_{l}$, $R$ by $R_{l}$ and $\tilde r$ by $r_{l}$.

The far-field pattern for this multiple-obstacle configuration in the half-plane can be obtained from the straightforward generalization of the formula (\ref{ff1}) previously derived in Section \ref{FFPSection}, which is given by
\begin{eqnarray}
u_{\infty}(\vartheta) = \sum_{j=1}^{J} F[v_{j}](\vartheta) + \sum_{j=1}^{J} \tilde{F} S [v_{j}](\vartheta), \label{ffm}
\end{eqnarray}
where $F[v_{j}](\vartheta)$ is given by (\ref{defF}) with $v$ replaced by $v_{j}$,
$\textbf{b}$ by $\textbf{b}_{j}$, and $R$ by $R_{j}$. Similarly, $\tilde{F} S [v_{j}](\vartheta)$ is given by (\ref{defFt}) with $\tilde{v}$ replaced by $S [v_{j}]$,
$\tilde{\textbf{b}}$ by $\tilde{\textbf{b}}_{j}$, and $R$ by $R_{j}$.

%%%%%%%%%%%%%%%%%%%%%%%%%%%%%%%%%%%%%%%%%%%%%%%%%%%%%%%%%%%%%%%%%%%%
%%   S E C T I O N
%%%%%%%%%%%%%%%%%%%%%%%%%%%%%%%%%%%%%%%%%%%%%%%%%%%%%%%%%%%%%%%%%%%%
\section{Numerical method} \label{GridGenFDMSection}
As mentioned in the introduction, the BVP (\ref{BVPDtNNM1})-(\ref{BVPDtNNM4}) defined in the relatively small bounded region
$\Omega_{\text{int}} = \cup_{j=1}^{J} \Omega_{j}$ is numerically solved
using a finite difference method which closely follows the technique introduced in \cite{JCP2010}.

Each sub-domain $\Omega_j$ is an annular
region with arbitrary piecewise smooth inner boundary
$\mathcal{C}_j$ and circular outer boundary $\mathcal{B}_j$. Local
boundary-conforming curvilinear coordinate systems are defined for
each sub-domain $\clo{\Omega}_{j}$. Thus, the coordinate lines
conform to the corresponding obstacle bounding curve
$\mathcal{C}_{j}$ and to the outer circle
$\mathcal{B}_{j}$.
The coordinates of a point in the bounded sub-domain
$\clo{\Omega}_{j}$ can be denoted as
$$\textbf{x}(\xi^{j},\eta^{j}) =
(x(\xi^{j},\eta^{j}),y(\xi^{j},\eta^{j})), \quad \mbox{for}\quad j=1\dots J.$$ These local systems of coordinates are
independent from each other. Thus, from now and on, we will ignore the superscript ``$j$" and will continue our description as if it were only one scatterer.

Our approach requires robust
supporting grids for the accuracy of the approximate solutions.
Therefore, there is a need to obtain appropriate grids for
scatterers of complexly shaped geometry. As done in previous work \cite{ViAcMATCOM,JCAM2009,JCP2010}, the approach adopted here consists of generating structured elliptic grids
numerically. This is accomplished through a transformation $T$ which establishes a relationship between generalized coordinates $(\xi,\eta)$ in a computational rectangular domain and the cartesian coordinates $(x,y)$ in a complexly shaped physical region.

A common practice in elliptic grid generation is
to implicitly define the transformation $T$ as the numerical
solution to a Dirichlet boundary value problem governed by quasi-linear elliptic system of partial differential equations. Excellent descriptions of this technique can be found in \cite{Steinberg,Handbook}. In particular in \cite{JCP2010}, we introduced the following system,
\begin{eqnarray}
&& \alpha x_{\xi \xi }-2\beta x_{\xi \eta}+\gamma x_{\eta \eta } +
\frac{1}{2} \alpha_{\xi}x_{\xi} + \frac{1}{2}
\gamma_{\eta}x_{\eta} = 0, \label{Elliptic1} \\
&& \alpha y_{\xi \xi }-2\beta y_{\xi \eta}+\gamma y_{\eta \eta } +
\frac{1}{2} \alpha_{\xi}y_{\xi} + \frac{1}{2} \gamma_{\eta}y_{\eta} = 0.
\label{Elliptic2}
\end{eqnarray}
where $\alpha = x_{\eta}^{2} + y_{\eta}^2$, $\beta = x_{\xi} x_{\eta} + y_{\xi} y_{\eta} $ and $\gamma = x_{\xi}^{2} + y_{\xi}^2$.
After imposing Dirichlet boundary conditions on the boundaries of the annular regions $\Omega_{\text{int}}$, the resulting BVP for this quasi-linear elliptic system is numerically solved and a grid is generated.
As shown in \cite{JCP2010} for annular regions with circular boundaries, the well-known polar coordinates, $x(\xi,\eta) = \eta \cos \xi$ and $y(\xi,\eta) = \eta \sin \xi$, exactly solve the elliptic system (\ref{Elliptic1})-(\ref{Elliptic2}).
A natural expectation is that grids numerically generated from
(\ref{Elliptic1})-(\ref{Elliptic2}) for arbitrarily shaped annular
regions preserve the good properties that polar grids have for
circular domains. Indeed, these grids are smooth, non-self-overlapping, and enjoy nearly uniform distribution of grid points as shown in Fig. \ref{Grid}.

\begin{figure}[!ht]
\begin{center}
\includegraphics[width=0.8\textwidth]{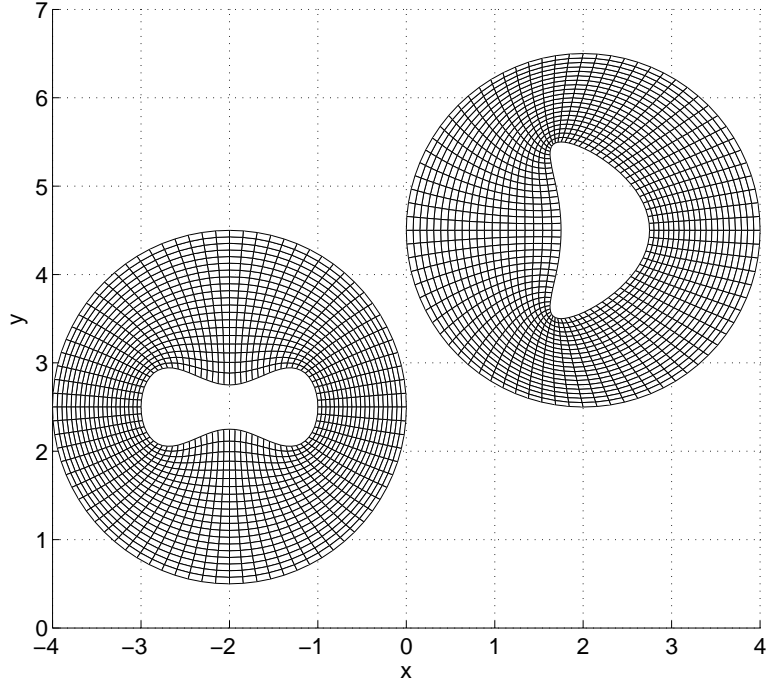}
\caption{Elliptic grids for a complexly shaped scatterers.}
\label{Grid}
\end{center}
\end{figure}

Before employing the elliptic grids in the numerical solution of the BVP (\ref{BVPDtNNM1})-(\ref{BVPDtNNM4}), the governing equations need to be written in terms of the new coordinates obtained form the grid generation process. For instance,
the Helmholtz equation in the \emph{elliptic} coordinates
generated by the system (\ref{Elliptic1})-(\ref{Elliptic2}) is given by
\begin{eqnarray}
\frac{1}{J^2} \Bigg[\alpha u_{\xi \xi
}-2\beta u_{\xi \eta}+\gamma u_{\eta \eta } +
\frac{1}{2}
\Big(\alpha_{\xi}\, u_{\xi}+\gamma_{\eta}\, u_{\eta}\Big)\Bigg]
+ k^2 u = 0, \label{HelmElliptic}
\end{eqnarray}
where $J= x_{\xi} y_{\eta} - x_{\eta} y_{\xi}$ is the jacobian of the transformation $T$.

The finite difference method employed in this work is based on a second-order discretization of all the derivatives present in both the grid generation system (\ref{Elliptic1})-(\ref{Elliptic2}) and in the BVP (\ref{BVPDtNNM1})-(\ref{BVPDtNNM4}) written in elliptic coordinates.
Then, numerical solutions are obtained inside the relatively small annular region enclosing the obstacle and excluding the plane boundary $\Gamma$. This has a major advantage. The number of grid points is greatly reduced compared with those obtained from enclosing the scatterer and a portion of the axis $\Gamma$ with a single large artificial boundary. Also, the grid generation process is much simpler because it is independently performed over small regions with a single obstacle. This is particularly true for the structured grids. Finite difference methods are attractive for the simulation of wave
phenomena due to their simplicity. Their application on complex geometries is facilitated by the construction of smooth boundary-conforming grids. For a detailed implementation of the proposed numerical technique the reader is referred to \cite{JCP2010}. Since the half-plane multiple-DtN boundary condition can be easily incorporated into the variational formulation of BVP (\ref{BVPDtNNM1})-(\ref{BVPDtNNM4}), the finite element numerical method can be naturally combined with this nonreflecting condition, as previously done in \cite{Keller01,Grote-Keller01,Givoli1999,Grote01}.

%%%%%%%%%%%%%%%%%%%%%%%%%%%%%%%%%%%%%%%%%%%%%%%%%%%%%%%%%%%%%%%%%%%
%%   S E C T I O N
%%%%%%%%%%%%%%%%%%%%%%%%%%%%%%%%%%%%%%%%%%%%%%%%%%%%%%%%%%%%%%%%%%%%
\section{Numerical examples} \label{NumericalExamples}
In this section we present some numerical results obtained from the application of the numerical method described in the previous section. The approximations were obtained by
truncating the series of the multiple-DtN operators $M$, $P$, $\tilde{P}$, $T$ and $\tilde{T}$ at $N = 30$ terms. To preserve the uniqueness of the solution, this value of $N$ was carefully chosen to satisfy the condition $N\ge \max (kR_l),$ for $l=1,\dots J$, as indicated in \cite{Grote01}.

Here our numerical method is validated by comparing its numerical approximations against the exact solution for one and for two circular obstacles embedded in the upper half-plane with the Neumann condition (\ref{BVPsc3}) imposed on $\Gamma$. The analytical solutions of these scattering problems and their far-field patterns can be obtained using eigenfunction expansions, as described in \cite[Chapter 4]{MartinBook} for two and four circular obstacles symmetrically located about $\Gamma$.

\subsection{One and two circular obstacles} \label{CircularObs}
First, we consider a single acoustically soft circular obstacle of radius $a=1$ with center at $\textbf{c} = (0,2)$. The artificial boundary $\mathcal{B}$ is a circle
of radius $R = 1.5$ with center at $\textbf{b} = \textbf{c}$. The direction of incidence is $\textbf{d}=(1,0)$, and the wavenumber $k=\pi$. Using the proposed technique, the far-field pattern for this problem $F_{\text{num}}$ is numerically computed.
The exact far-field pattern $F_{\text{exact}}$ is also obtained using eigenfunction expansions, as described in \cite[Chapter 4]{MartinBook}, for two cylinders located symmetrically about $\Gamma$. Figure \ref{FFPError1} displays the relative difference between $F_{\text{num}}(\vartheta)$ and $ F_{\text{exact}}(\vartheta)$ as a function of the angle $\vartheta \in [0,\pi]$. To obtain this plot, we employed a grid of size $200 \times 20$ to discretize the vicinity of the obstacle. With this moderate grid size, the numerical far-field pattern shows excellent agreement with the exact one with a relative error less than $0.17 \%$.

\begin{figure}[!ht]
\begin{center}
\includegraphics[width=0.8\textwidth]{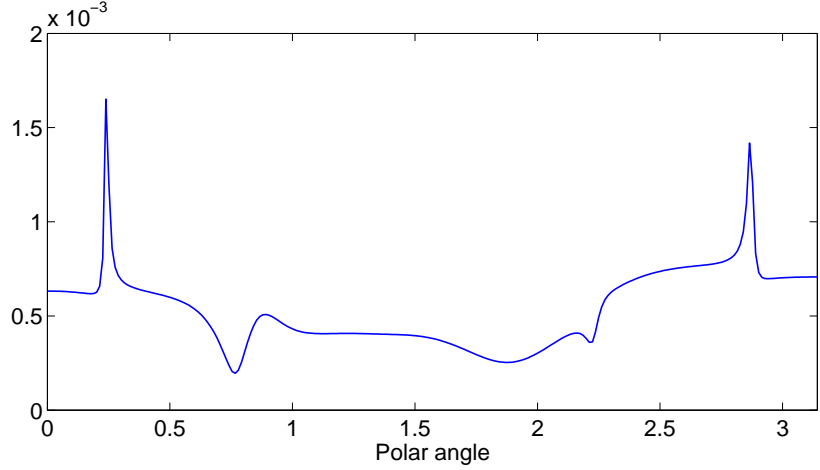}
\caption{Relative error of the numerical far-field pattern as a function of the polar angle for one circular obstacle embedded in the half-plane.}
\label{FFPError1}
\end{center}
\end{figure}

The other benchmark problem considered consists of the scattering from two acoustically soft circular obstacles located in the upper-half plane. The first obstacle is a circle of radius $a_{1}=1$ with center at $\textbf{c}_{1} = (-2,3)$ and the second has radius $a_{2}=1$ and center at $\textbf{c}_{2} = (2,2)$. The artificial boundaries $\mathcal{B}_{1}$ and $\mathcal{B}_{2}$ are circles of radius $R_{1} = R_{2} = 1.5$ with centers at $\textbf{b}_{1} = \textbf{c}_{1}$ and $\textbf{b}_{2} = \textbf{c}_{2}$, respectively. Here again, the direction of incidence is $\textbf{d} = (1,0)$ and the wavenumber $k=\pi$. Figure \ref{FFPError2} displays the relative difference between the exact and numerical far-field patterns for the scattering from two circular obstacle embedded in the upper half-plane. The plot was obtained for the grid size $200 \times 20$. Again even for this moderate grid size, excellent agreement between the approximate and the exact far-field patterns is observed with a relative error below $0.84 \%$.

\begin{figure}[!ht]
\begin{center}
\includegraphics[width=0.8\textwidth]{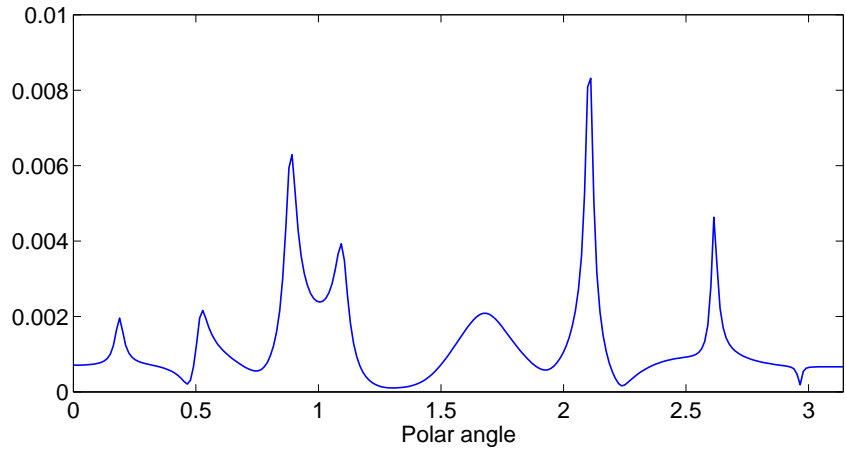}
\caption{Relative error of the numerical far-field pattern as a function of the polar angle for two circular obstacles embedded in the half-plane.}
\label{FFPError2}
\end{center}
\end{figure}

In order to verify the order of convergence in both the one-obstacle and the two-obstacle problems, a sequence of numerical
tests was made for increasingly finer grids conforming to each
circular obstacle. The second order convergence of the FDM is easily verified from the results shown in Table \ref{Table1}.

\begin{table}[!htb]
\small
\begin{center}
\caption{Maximal relative error for the numerical far-field patterns.}
\label{Table1}
\begin{tabular}{c c c c}
\hline
Grid Size    & One Circle & Two Circles  \\
\hline
$100\times10$ & $6.63\times 10^{-3}$ & $3.33\times 10^{-2}$  \\
$150\times15$ & $2.94\times 10^{-3}$ & $1.48\times 10^{-2}$  \\
$200\times20$ & $1.65\times 10^{-3}$ & $8.32\times 10^{-3}$  \\
$250\times25$ & $1.06\times 10^{-3}$ & $5.32\times 10^{-3}$  \\
$300\times30$ & $7.33\times 10^{-4}$ & $3.69\times 10^{-3}$  \\
\hline
\end{tabular}
\end{center}
\end{table}

\subsection{Complexly shaped obstacles} \label{ComplexObs}
To show the capability of our technique to deal with obstacles of arbitrary shape, we present some results for complexly shaped obstacles. This experiment considers two obstacles embedded in the upper half-plane. They are shaped like a \emph{peanut} and a \emph{kite} whose parametric equations are given below. The direction of incidence is $\textbf{d}=(\frac{1}{\sqrt{2}},-\frac{1}{\sqrt{2}})$ and wavenumber $k=\pi$. This particular example illustrates the multiple scattering interaction as the waves bounce back and forth between the two obstacles and the plane boundary $\Gamma$. The results from three experiments are shown in Figure \ref{ComplexScatt}. The one on the top left corresponds to the Neumann condition imposed on the plane boundary $\Gamma$, whereas the plot on the top right corresponds to the Dirichlet condition on $\Gamma$. As seen in Figure \ref{ComplexScatt}, the artificial boundaries $\mathcal{B}_{1}$ and $\mathcal{B}_{2}$ are circles of radius $R_{1} = R_{2} = 2.0$. The nonlocal DtN condition allows us to take a smaller radius bringing the artificial boundaries arbitrarily close to the obstacles. However, these plots help us visualize the wave fields and the influence from the plane boundary $\Gamma$ and the type of boundary condition imposed thereon. In fact, in order to better appreciate the effect of correctly including the plane boundary $\Gamma$ in the problem, we also display the scattering event in the absence of this infinite boundary as shown in the third plot (bottom) of Figure \ref{ComplexScatt}.
\begin{itemize}
\item[-]{Peanut: $x(t) = \frac{1}{11}(10+ \cos 2t) \cos t - 2$, $~y(t) = \frac{1}{16}(10+6 \cos 2t) \sin t + \frac{5}{2}$, $~t \in [0,2\pi]$.}
\item[-]{Kite: $x(t) = \frac{1}{4}(2 \cos t + \cos 2t) \cos t + 2$, $~y(t) = \sin t + \frac{9}{2}$, $~t \in [0,2\pi]$.}
\end{itemize}

\begin{figure}[]
\begin{center}
\includegraphics[width=0.5\textwidth]{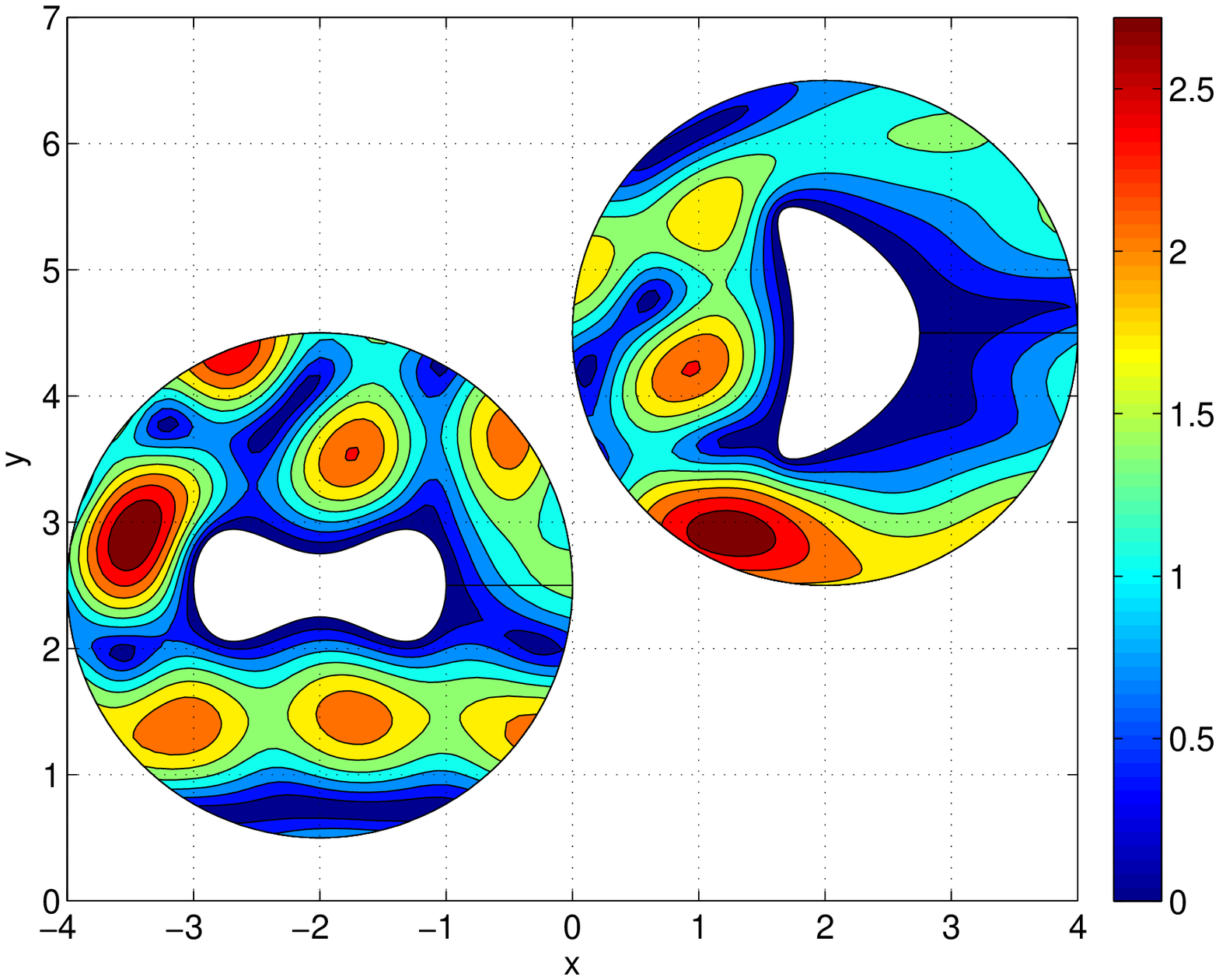} \hspace{-0.2in}
\includegraphics[width=0.5\textwidth]{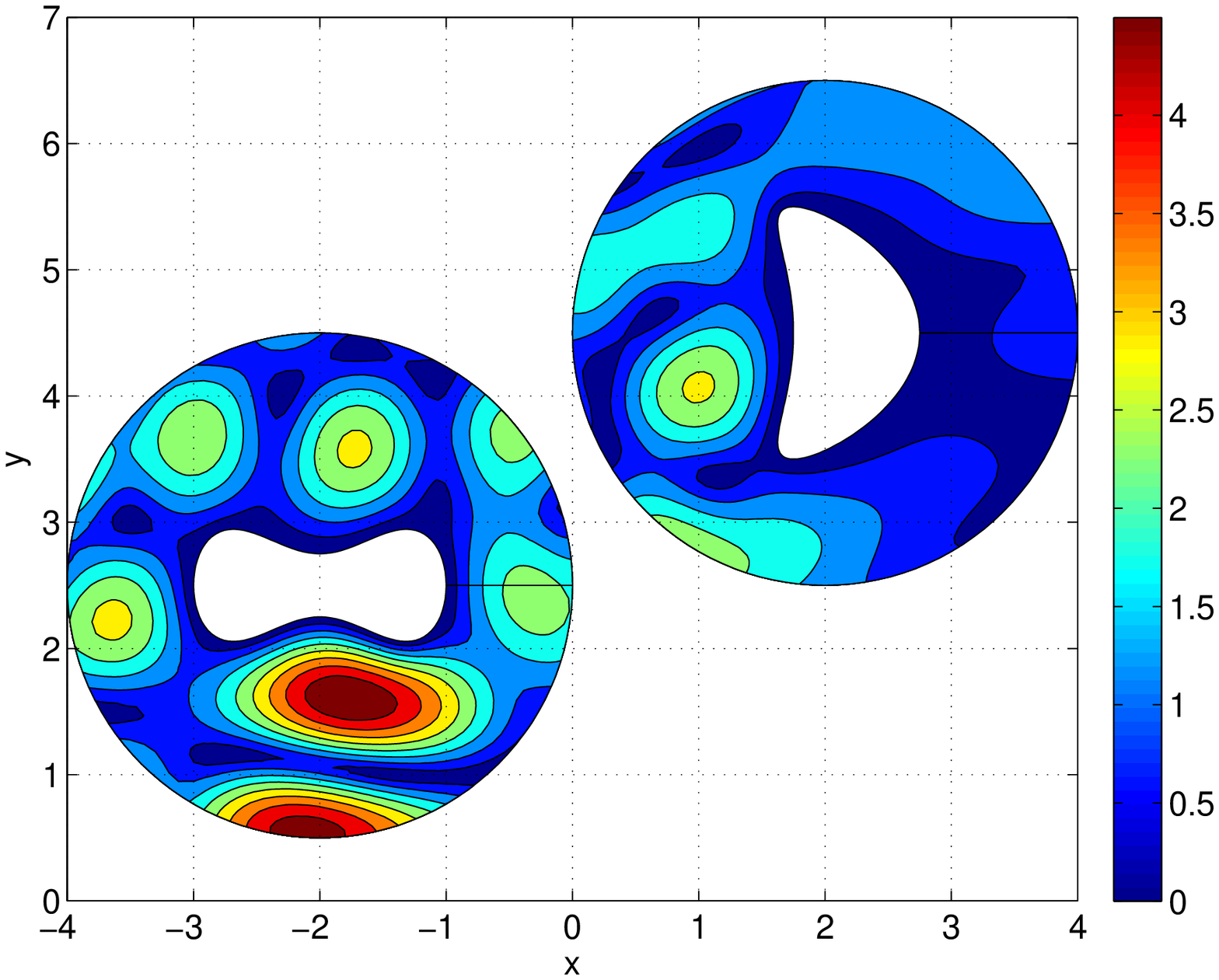} \vspace{-0.2in}
\includegraphics[width=0.5\textwidth]{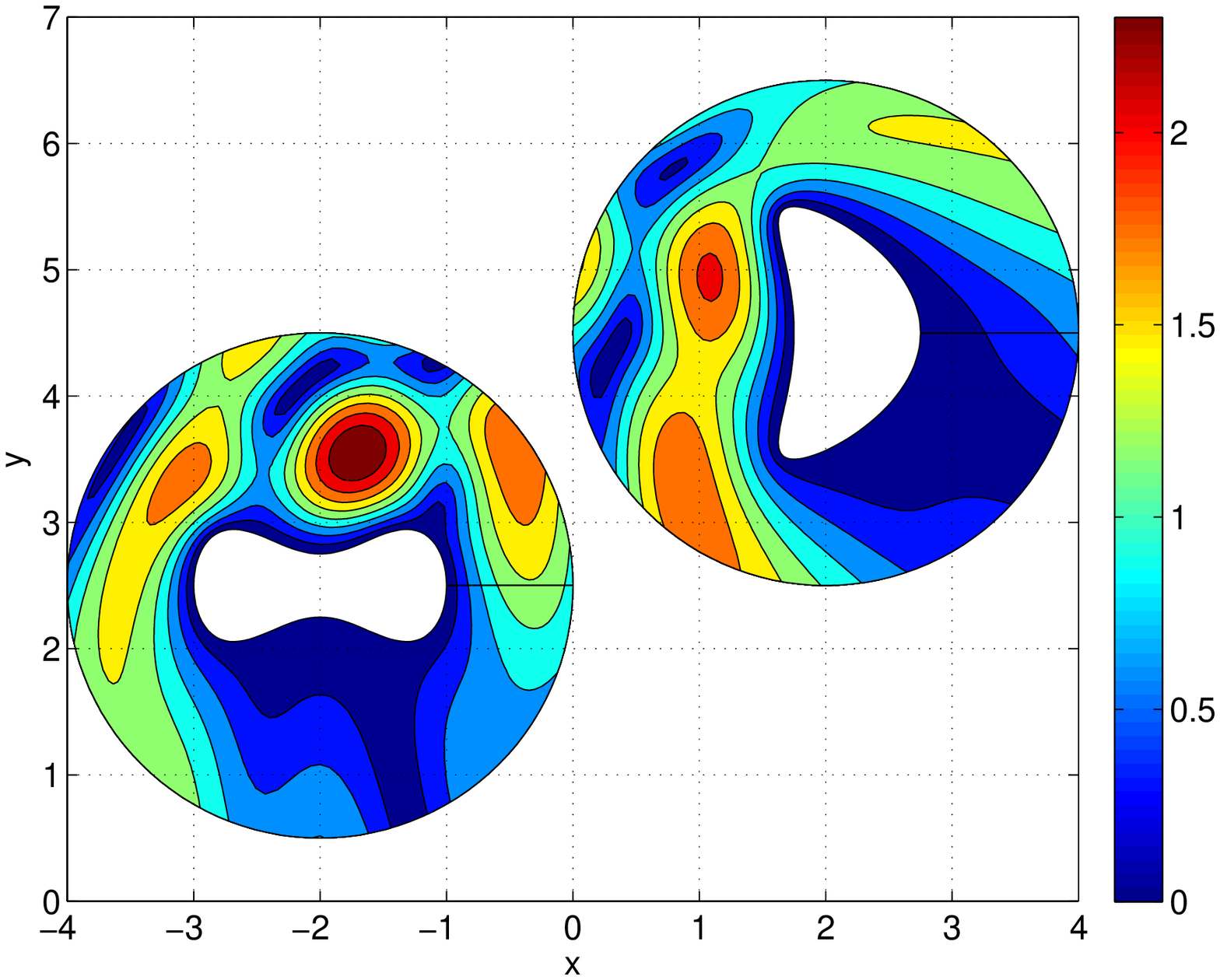}
\caption{Amplitude of the total field. The plots on the top correspond to the scattering in the presence of the plane boundary $\Gamma$ with a Neumann condition (top left) and Dirichlet condition (top right), respectively. The third plot (bottom) displays the scattering event in the absence of the plane boundary $\Gamma$, i.e., in the full-plane.}
\label{ComplexScatt}
\end{center}
\end{figure}

%%%%%%%%%%%%%%%%%%%%%%%%%%%%%%%%%%%%%%%%%%%%%%%%%%%%%%%%%%%%%%%%%%%%
%%   S E C T I O N
%%%%%%%%%%%%%%%%%%%%%%%%%%%%%%%%%%%%%%%%%%%%%%%%%%%%%%%%%%%%%%%%%%%%

\section{Conclusions and future work}
A novel DtN map for scattering from a single or multiple obstacles embedded in the half-plane has been constructed. The boundary condition resulting from this map, defined by (\ref{BVPDtNN3})-(\ref{BVPDtNN4}) for a single obstacle or by (\ref{BVPDtNNM3})-(\ref{BVPDtNNM4}) for multiple obstacles, has been called the \textit{half-plane multiple-DtN boundary condition}. This condition simultaneously deals with the unboundedness of the media and the presence of a plane boundary $\Gamma$ which can be either acoustically hard or soft. The new DtN map is derived combining the full-plane multiple-DtN map and the method of images.

The application to half-plane multiple scattering problems consists of enclosing each obstacle with a circular artificial boundary on which the half-plane DtN boundary condition is imposed. This condition not only allows waves to leave the computational sub-domains without spurious reflections, but it also communicates the field data from one sub-domain to another. In addition, the half-plane multiple-DtN boundary condition successfully deals with the multiple scattering interactions which are unavoidably present in wave propagation problems in semi-infinite media even if only one obstacle is considered.

Since the scattered field in the common exterior region $\Omega_{\text{ext}},$ to all circles  $\mathcal{B}_j,$ has an analytical representation, the artificial boundaries can be placed arbitrarily close to the obstacles. There are several advantages in doing this. First, the size of each one of the sub-domains $\Omega_j$ where the numerical solution is calculated is relatively small and each of the $\Omega_j$ does not include
any portion of the plane boundary $\Gamma$. Secondly, simulation for moderately high values of the wavenumber $k$ can be performed. This would be much more computationally expensive using a single semi-circular artificial boundary as is currently done.

The proposed method is especially suited for the simulation of multiple scattering interactions by distant obstacles embedded in semi-infinite media. For instance, it should have a significant impact on active sonar where sound interaction between two or more submarines inside the ocean takes place \cite{UrickBook1983,Schenck1986}. This is due to the fact that the distance between the submarines might be several orders of magnitude larger than the length of the submarines. A similar situation is encountered in the simulation of acoustic and electromagnetic wave interaction between one or more airplanes and the ground.

Finally, we offer the following brief description of possible improvements and extensions of the present work. Some of them are clearly feasible but others may require a good deal of research to be corroborated. 
\begin{enumerate}
\item{The formulation of the half-plane multiple-DtN boundary condition can also be generalized to three space dimensions. In fact, Lemmas \ref{Lemma1}-\ref{LemmaSymmetry} and Theorems \ref{TheoremDecomp}-\ref{TheoremEquiv} are easily extendable to their three-dimensional counterparts. However, the expressions (\ref{defP})-(\ref{defTt}) defining the multiple-DtN operators $M$, $P$ and $T$ become more involved and harder to implement. The derivation of this multiple-DtN condition can also be generalized to elliptically shaped artificial boundaries \cite{Grote-Keller01}. This will provide more flexibility in the selection of the computational region to be discretized.}

\item{It should also be possible to combine semi-circular and circular artificial boundaries as needed. For instance, if one of the several obstacles is too close to the plane boundary $\Gamma$ to fit an enclosing circle, then a semi-circular artificial boundary would be preferred. A DtN condition  \cite{GivoliBook,GivVig1993} may be imposed on the semi-circle. Since an analytic expression for the purely outgoing field is available then appropriate multiple-DtN operators $M$, $P$ and $T$ may be formulated.}

\item{The localization of the multiple-DtN operators. In order to alliviate the numerical burden associated with a nonlocal boundary condition, some approaches have been devised to localize the single-DtN operator $M$ \cite{GivoliBook,Givoli1999,GivPat98}. The goal would be to obtain similar localizations of the other multiple-DtN operators $P$ and $T$. This localization should lead to sparse matrices from the FEM and FDM for multiple scattering problems in the full- and the half-plane.}

\item{The derivation of analogous half-space multiple-DtN boundary conditions for time dependent wave problems. This extension seems to be accomplishable since full-space time dependent multiple-DtN conditions have already been developed \cite{Grote02,Grote03}.}

\item{The appropriate extension to other linear partial differential equations. The two main ingredients of this work are the full-space multiple-DtN condition and the method of images. Hence, in principle, the half-space multiple-DtN condition may be formulated for BVPs where the above ingredients are known to hold. This is the case for the Laplace equation governing physical phenomena such as electrostatics, potential flow, steady heat transfer, etc. For Maxwell equations, exact nonreflecting conditions were formulated in \cite{GroKell98,Grote06} for single scattering, and the method of images holds for electrodynamics in the presence of a perfectly conducting plane boundary \cite{Low2011}.}

\item{Other boundary conditions on the plane boundary $\Gamma$. The method of images is naturally associated with symmetric or anti-symmetric boundary conditions such as the homogeneous Dirichlet or Neumann conditions. Unfortunately, very important conditions do not lead to such symmetries. Some of these are the traction-free condition in elasticity, and the impedance condition in acoustics or electrodynamics. However, some effort has been made to construct appropriate image methods for these boundary conditions. See \cite{Taral05,Chandler10,MaLin01} and references therein. At this point, it is not clear whether the referred image methods are compatible with the formulation of the DtN condition.}
\end{enumerate}

%%%%%%%%%%%%%%%%%%%%%%%%%%%%%%%%%%%%%%%%%%%%%%%%%%%%%%%%%%%%%%%%%%%%
%   A C K N O W L E D G M E N T S
%%%%%%%%%%%%%%%%%%%%%%%%%%%%%%%%%%%%%%%%%%%%%%%%%%%%%%%%%%%%%%%%%%%%

%\section*{Acknowledgments}
%The authors would like to thank the anonymous referees for their
%most constructive suggestions which certainly improved the quality
%of the manuscript. The first author's work was supported by the
%CHIRP grant of the College of Physical and Mathematical Sciences of Brigham Young
%University.

%%%%%%%%%%%%%%%%%%%%%%%%%%%%%%%%%%%%%%%%%%%%%%%%%%%%%%%%%%%%%%%%%%%%
%   B I B L I O G R A P H Y
%%%%%%%%%%%%%%%%%%%%%%%%%%%%%%%%%%%%%%%%%%%%%%%%%%%%%%%%%%%%%%%%%%%%

\bibliographystyle{elsarticle-num}
\bibliography{AcoBib}

\end{document}